\documentclass[sigconf]{acmart}

\AtBeginDocument{%
  \providecommand\BibTeX{{%
    \normalfont B\kern-0.5em{\scshape i\kern-0.25em b}\kern-0.8em\TeX}}}

\copyrightyear{2022}
\acmYear{2022}
\setcopyright{acmlicensed}\acmConference[SPAA '22]{Proceedings of the 34th ACM Symposium on Parallelism in Algorithms and Architectures}{July 11--14, 2022}{Philadelphia, PA, USA}
\acmBooktitle{Proceedings of the 34th ACM Symposium on Parallelism in Algorithms and Architectures (SPAA '22), July 11--14, 2022, Philadelphia, PA, USA}
\acmPrice{15.00}
\acmDOI{10.1145/3490148.3538580}
\acmISBN{978-1-4503-9146-7/22/07}

\newcommand{\defn}[1]       {{\textit{\textbf{\boldmath #1}}}}

\newtheorem{claim}{Claim}

\begin{document}

\title{Bamboo Trimming Revisited: Simple Algorithms Can Do Well Too}

\author{John Kuszmaul}
\affiliation{%
\institution{Yale University}
\city{New Haven}
\state{Connecticut}
\country{USA}
}
\email{john.kuszmaul@gmail.com}

\begin{abstract}
The bamboo trimming problem considers $n$ bamboo with growth rates $h_1, h_2, \ldots, h_n$ satisfying $\sum_i h_i = 1$. During a given unit of time, each bamboo grows by $h_i$, and then the bamboo-trimming algorithm gets to trim one of the bamboo back down to height zero. The goal is to minimize the height of the tallest bamboo, also known as the backlog. The bamboo trimming problem is closely related to many scheduling problems, and can be viewed as a variation of the widely-studied fixed-rate cup game, but with constant-factor resource augmentation. 

Past work has given sophisticated pinwheel algorithms that achieve the optimal backlog of 2 in the bamboo trimming problem. It remained an open question, however, whether there exists a \emph{simple} algorithm with the same guarantee---recent work has devoted considerable theoretical and experimental effort to answering this question. Two algorithms, in particular, have appeared as natural candidates: the \textbf{Reduce-Max} algorithm (which always cuts the tallest bamboo) and the \textbf{Reduce-Fastest}$(x)$ algorithm (which cuts the fastest-growing bamboo out of those that have at least some height $x$). It is conjectured that \textbf{Reduce-Max} and \textbf{Reduce-Fastest}$(1)$ both achieve backlog 2.

This paper improves the bounds for both \textbf{Reduce-Fastest} and \textbf{Reduce-Max}. Among other results, we show that the \emph{exact optimal} backlog for \textbf{Reduce-Fastest}$(x)$ is $x + 1$ for all $x \ge 2$ (this proves a conjecture of D’Emidio, Di Stefano, and Navarra in the case of $x = 2$), and we show that \textbf{Reduce-Fastest}$(1)$ \emph{does not} achieve backlog 2 (this disproves a conjecture of D’Emidio, Di Stefano, and Navarra). 

Finally, we show that there is a different algorithm, which we call the \textbf{Deadline-Driven} Strategy, that is both very simple and achieves the optimal backlog of 2. This resolves the question as to whether there exists a simple worst-case optimal algorithm for the bamboo trimming problem.  
\end{abstract}









\maketitle


\section{Introduction}

\label{sec:intro}

A classic scheduling problem is the so-called \defn{cup game},
which is a two-player game that takes place on $n$ cups.
In each step of the game, the filler player distributes $1$ unit of water among the cups arbitrarily;
the emptier player then selects a single cup and removes up to $1$ unit of water from that cup.
The emptier's goal is to minimize the backlog of the system, which is defined to be the amount of water in the fullest cup.

The cup game was first introduced in the late 1960s  \cite{Liu69, LiuLa73}, and has been studied in many different forms \cite{BaruahCoPl96,GasieniecKlLe17,BaruahGe95,LitmanMo11,LitmanMo05,MoirRa99,BarNi02,GuanYi12,Liu69, LiuLa73,DietzRa91, BenderFaKu19, Kuszmaul20, AdlerBeFr03, DietzSl87,LitmanMo09}. 
The game has found extensive applications in areas such as
processor scheduling  \cite{BaruahCoPl96,GasieniecKlLe17,BaruahGe95,LitmanMo11,LitmanMo05,MoirRa99,BarNi02,GuanYi12,Liu69, LiuLa73, AdlerBeFr03, LitmanMo09, DietzRa91,BenderFaKu19, Kuszmaul20, tailsize, BenderKu20}, network-switch buffer management \cite{Goldwasser10,AzarLi06,RosenblumGoTa04,Gail93}, quality of service guarantees \cite{BaruahCoPl96,AdlerBeFr03,LitmanMo09}, and data-structure deamortization \cite{AmirFaId95,DietzRa91,DietzSl87,AmirFr14,Mortensen03,GoodrichPa13,FischerGa15,Kopelowitz12,BenderDaFa20}. See \cite{Kuszmaul20} for a detailed discussion of the related work.

Perhaps the most natural emptying algorithm is the \textbf{Reduce-Max} algorithm, which always empties from the fullest cup.
\textbf{Reduce-Max} achieves an asymptotically optimal backlog of  $O (\log n) $ \cite{AdlerBeFr03, DietzSl87}. In fact, in addition to being asymptotically optimal,
\textbf{Reduce-Max} is known to be \emph{exactly} optimal---no other algorithm can do better, even by an additive constant \cite{AdlerBeFr03}. 

An important special case of the cup game is the setting where the filler's behavior is the same on every step, also known as the fixed-rate cup game.
Whereas the optimal backlog in the variable-rate cup game is $O(\log n)$, the optimal backlog in the fixed-rate cup game is $O(1)$ \cite{BaruahCoPl96,GasieniecKlLe17,BaruahGe95,LitmanMo11,LitmanMo05,MoirRa99,BarNi02,GuanYi12,Liu69,
  LiuLa73}.
Perhaps surprisingly, though, the \textbf{Reduce-Max} algorithm is no longer optimal (or even asymptotically optimal!).
In fact, the algorithm still allows for backlog $\Omega(\log n)$ in the fixed-rate setting \cite{AdlerBeFr03}.

Recent work has identified a potential path to redemption for the \textbf{Reduce-Max} algorithm, however.
Bil{\`o}, Gual{\`a}, Leucci, Proietti, and Scornavacca \cite{BiloGuLe20} showed that, if the emptier is given resource augmentation over the filler, meaning that the emptier is 
permitted to fully empty a cup on each step rather than removing just a single unit of water,
then the backlog achieved by the \textbf{Reduce-Max} algorithm becomes $O (1) $.
Note that, since the backlog is constant,
the resource augmentation never results in the emptier removing more than $O(1)$ units of water at a time.

Although there is a long history of studying resource-augmented variants of the cup game \cite{BenderFaKu19, LitmanMo09, DietzRa91, DietzSl87},
it is only relatively recently that researchers have begun to study the resource-augmented fixed-rate
version of the game \cite{BiloGuLe20, GasieniecKlLe17, DEmidioStNa19}. These papers have dubbed the problem as the 
\defn{Bamboo Garden Trimming Problem}, based on the following (rather creative) problem interpretation. A robotic panda gardener
is responsible for maintaining a bamboo garden. The garden consists of $n$ bamboo $b_1, \ldots, b_n$ with corresponding growth-rates $h_1 \geq \ldots \geq h_n$
 satisfying $\sum_{i=1}^n h_i = 1$. Each bamboo $b_i$ starts at height $0$ and grows at a steady rate of $h_i$ every time unit. 
At the end of each time unit, the player chooses one of the bamboo and chops this bamboo down to height $0$. 
The goal, of course, is to achieve the smallest possible backlog, which is the height of the tallest bamboo.

It is known that no bamboo trimming algorithm can guarantee a backlog less than $2$, as it is possible to achieve backlog at least 
$2 - 2\varepsilon$ against any bamboo trimming algorithm with two bamboo that have fill rates $1 - \varepsilon$ and $\varepsilon$ \cite{BiloGuLe20}.
 Recent work has yielded complex pinwheel algorithms \cite{GasieniecKlLe17} that achieve backlog $2$, and are thus optimal in terms of the worst-case backlog; there has also been effort to extend the guarantees of these algorithms to achieve strong competitive ratios for cases where backlog less than 2 is possible \cite{pinwheel2, pinwheel3}.

\subsection{The quest for a simple optimal bamboo-trimming algorithm}
The relative complicatedness of the known pinwheel algorithms has sparked a great deal of interest
in the question as to whether there exists some \emph{simple} algorithm that achieves the optimal backlog of 2.
It would be especially interesting if \textbf{Reduce-Max} were to achieve backlog 2, since this would mean that the algorithm is optimal
for both bamboo-trimming and the standard cup game. Currently, the best known bound for \textbf{Reduce-Max} is a backlog of 9 \cite{BiloGuLe20}.
Experimental work \cite{DEmidioStNa19} has found that \textbf{Reduce-Max} does, in fact, seem to achieve a backlog of 2, however, leading the authors to pose
a backlog of 2 as a conjecture.

Another algorithm family that has been studied for its simplicity is \textbf{Reduce-Fastest}($x$). This algorithm trims down the fastest-growing bamboo out of those
 that have height at least $x$ at the end of each time unit. Initial study proved that \textbf{Reduce-Fastest}($2$) achieves backlog at most $4$ \cite{GasieniecKlLe17},
 and further work demonstrated that \textbf{Reduce-Fastest}($x$) achieves backlog at most
 $$\max\left(x + \frac{x^2}{4(x-1)}, \frac{1}{2} + x + \frac{x^2}{4(x - 1/2)}\right)$$ for all $x > 1$, which yields a bound of $19/6$ at $x = 2$ \cite{BiloGuLe20} (and is $\ge 1.25x$ for all $x > 1$).
Extensive computer experimentation \cite{DEmidioStNa19} suggests that this bound of $19/6$ is still not optimal, and has led researchers to conjecture that \textbf{Reduce-Fastest}($2$) actually achieves a backlog of $3$.
Based on the same experiments, the authors further conjecture that \textbf{Reduce-Fastest}($1$) achieves the optimal backlog of $2$ \cite{DEmidioStNa19}. However, as of now, no theoretical bounds on the backlog of \textbf{Reduce-Fastest}($1$) are known.

\subsection{Our Results}
Our first result is an improved bound on the backlog of \textbf{Reduce-Max} for bamboo trimming.
We prove that \textbf{Reduce-Max} achieves a backlog of $4$, which narrows the gap between the upper and lower bounds from $7$ to $2$.
At a technical level, our bound relies on a novel potential-function argument; we believe that this argument may be of independent interest
as a tool that could help in future analyses of similar problems.

Our second set of results analyze \textbf{Reduce-Fastest($x$)} for different values of $x$. We are able to prove that \textbf{Reduce-Fastest$(x)$} achieves 
backlog $x+1$ for all $x \geq 2$, and we give a matching lower bound showing that this analysis is tight. This is the first time that a tight analysis has been achieved
for \textbf{Reduce-Fastest}$(x)$ for any value of $x$. For $x = 2$, the result gives a backlog of 3, which resolves a conjecture of \cite{DEmidioStNa19}.
On the other hand, we disprove the conjecture of \cite{DEmidioStNa19} that \textbf{Reduce-Fastest}($1$) achieves backlog 2. 
Instead, we show that \textbf{Reduce-Fastest}($1$) allows for a backlog of $3 - \epsilon$ for any $\epsilon > 0$. 
More generally, we show that there is no $x$ for which \textbf{Reduce-Fastest}$(x)$ achieves a backlog of $2.01$,
meaning that \textbf{Reduce-Fastest}$(x)$ is no longer a candidate in the quest for a simple optimal algorithm.

Our final result is a simple algorithm, which we call the \defn{Deadline-Driven Strategy}, that 
does in fact achieve a backlog of $2$. The algorithm, which is based upon Liu and Layland's algorithm from the early 70s 
for a related scheduling problem \cite{LiuLa73}, maintains the simplicity associated with \textbf{Reduce-Max} and \textbf{Reduce-Fastest($x$)}, 
while also matching the backlog bounds of the more complicated pinwheel-based algorithms. 

The \textbf{Deadline-Driven Strategy} selects the bamboo that will soonest achieve height $2$ of the bamboo that have height at least $1$. 
From a scheduling standpoint, we can consider the time at which a bamboo achieves height $2$ to be its deadline. From this perspective, 
the \textbf{Deadline-Driven Strategy} is simply chopping down the bamboo with the closest deadline, while not considering very short bamboo with
 height less than $1$. The \textbf{Deadline-Driven Strategy} shares an intriguing relationship with \textbf{Reduce-Max}
 and \textbf{Reduce-Fastest($x$)}. The \textbf{Reduce-Max} strategy is concerned solely with the height of a bamboo, whereas the 
\textbf{Reduce-Fastest($x$)} strategy is concerned solely with the speed of a bamboo. \textbf{Reduce-Max} can be thought of as 
cutting down the bamboo that is closest to achieving height $2$ in terms of height (possibly selecting that bamboo that has furthest 
surpassed $2$ if the conjecture that \textbf{Reduce-Max} achieves backlog $2$ is false), and \textbf{Reduce-Fastest($x$)} cuts down
 the quickest bamboo that has surpassed some threshold $x$. The \textbf{Deadline-Driven Strategy} strikes a balance between these two ---
instead of cutting down the bamboo that is closest to $2$ in terms of distance or the bamboo with the maximum speed, it cuts down 
the bamboo that is closest to $2$ in terms of time. It's interesting that of these three simple bamboo trimming algorithms, one is
 concerned with distance, one with speed, and the third with time.

\subsection{The Relationship to the Multi-Processor Cup Game}

We observe that many of our results in this paper extend to the multi-processor fixed-rate cup flushing game, which is the analogous multi-processor version of the bamboo garden trimming problem.

In each step of the multi-processor cup game with $p$ processors, the filler places $p$ units of water into the cups, with no more than $1$ unit of water going to any cup. The emptier then removes water from each of $p$ cups -- the emptier corresponds to a $p$-processor machine. The multi-processor version of the bamboo garden trimming problem is the multi-processor fixed-rate cup flushing game, in which the player empties $p$ cups entirely instead of removing only $1$ unit of water from each of $p$ cups.

As noted by \cite{BenderFaKu19}, solutions to the single-processor version of the fixed-rate cup flushing game immediately yield solutions to the multi-processor version, since a time step in the multi-processor version can be modelled as a chunk of $p$ steps of the single-processor game in which the fill rates are reduced by a factor of $p$. Thus, an upper bound of $y$ on the backlog achieved by an algorithm in the single-processor version of the bamboo trimming problem immediately yields a corresponding algorithm that achieves backlog no more than $y+1$ in the $p$-processor version of the bamboo trimming problem. (A gap of $1$ is lost since cups emptied in the first step of a chunk of $p$ steps in the single-processor fixed-rate cup flushing game will not be emptied until the end of the corresponding time step in the corresponding multi-processor game, resulting in a backlog as much as $(p-1)/p$ units of water higher).

Thus we are able to show that the analagous version of \textbf{Reduce-Max} achieves backlog at most $5$ in the multi-processor bamboo trimming game, that $\textbf{Reduce-Fastest}(x)$ achieves backlog at most $x+2$ for all $x \geq 2$, and that the $\textbf{Deadline-Driven Strategy}$ achieves backlog at most $3$ for the multi-processor version of the bamboo trimming problem.

\section{Reduce-Max}

In this section, we prove the following theorem.

\begin{theorem}
The \textbf{Reduce-Max} algorithm limits the backlog to strictly less than $4$.
\label{thm:greedy}
\end{theorem}

Recall that we have $n$ bamboo $b_1, \ldots, b_n$ with corresponding growth rates $h_1 \geq h_2 \geq \cdots \geq h_n$. We denote the height of $b_i$ at time $t$, after the $t$-th cut, by $|b_i|_t$. For $i \in [n]$ and $t \in \mathbb{N} \cup \{0\}$, we define the \defn{volume function}
\[V(i, t) = \sum_{k=1}^i \min(2, |b_k|_t)\]
to be the function measuring the height at time $t$ of the bamboo with growth rates at least $h_i$, counting tall bamboo as having height at most $2$.

For $i \in [n]$ and $t \in \mathbb{N} \cup \{0\}$, we then define a potential function 
\[\Phi(i, t) = \sum_{\substack{k \in [i] \\ 2(k-1) < V(i, t)}} h_k \cdot \min(2, V(i, t) - 2(k-1)),\]
which is a weighted sum of $h_1, \ldots h_i$, where the multiplicative weights sum to $V(i, t)$ and are each at most $2$. The weights are distributed to maximize the sum by putting as much weight as possible on the earlier values of $k$. For example, if $V(i, t) = 7.25$, we would have $$\Phi(i, t) = 2 h_1 + 2 h_2 + 2 h_3 + 1.25 h_4.$$

We prove the following lemma by examining the behavior of our potential function $\Phi(i, t)$ over time.

\begin{lemma}
For all $i \in [n]$ and $t \in \mathbb{N} \cup \{0\}$,
\[ |b_i|_t \leq 4 - \Phi(i, t) \leq 4.\]
\label{lem:greedy}
\end{lemma}

\begin{proof}

We proceed by induction on time $t$. For the base case we consider $t = 0$, in which case $$|b_i|_t = 0 \leq 2 \leq 4 - \Phi(i, t)$$ for all $i \in [n]$. (Note that $0 \leq \Phi(i, t) \leq \sum_{k=1}^n 2  h_k \leq 2$ by the definition of $\Phi$.)

For the inductive step, we suppose that the theorem holds at $t$ for all $i$. We will then prove that the theorem also holds at time $t+1$ for all $i$.

Let $i \in [n]$. We know that $$|b_k|_t \leq 4 - \Phi(k, t)$$ for all $k \in [n]$ by the inductive hypothesis. Between time $t$ and $t+1$ each $b_k$ first grows by $h_k$, and then the tallest bamboo, some $b_j$, is cut down. We refer to the time between the bamboo growing and the tallest bamboo being cut down as the \defn{intermediate} step. We assume that $b_j$ has height at least $2$ during the intermediate step between $t$ and $t+1$, as otherwise the lemma trivially holds for all bamboo at time $t+1$ (none of the bamboo will even have height 2 at time $t+1$). We will complete the proof with $3$ cases. 

\paragraph{Case 1: \boldmath$j < i$.}

In this case, a quicker-growing bamboo was cut down in the stead of $b_i$.  

We know that $b_i$ has grown exactly $h_i$ units from time $t$ to time $t+1$:
\begin{equation} \label{eq:case1b_i}
|b_i|_{t+1} = |b_i|_t + h_i.
\end{equation}

We also know that the volume function satisfies $V(i, t+1) \leq V(i, t) - 1$ because the growth step adds at most $1$ unit of volume, and then a bamboo $b_j$ with intermediate height at least $2$ and with $j < i$ is cut down, which removes $2$ units of volume. That is, the volume function with respect to $i$ decreases by at least $1$ unit from time $t$ to $t+1$. Thus
\begin{equation} \label{eq:case1excess}
\Phi(i, t+1) \leq \Phi(i, t) - h_i.
\end{equation}
since this removed unit of volume would have been weighted by some growth-rate at least $h_i$ in the weighted sum $\Phi(i, t+1)$.

By the inductive hypothesis, we know that the lemma holds for time $t$, so we have
\[|b_i|_t \leq 4 - \Phi(i, t).\]
Substituting with Equations \eqref{eq:case1b_i} and \eqref{eq:case1excess} we have
\[|b_i|_{t+1} - h_i \leq 4 - (\Phi(i, t+1) + h_i)\]
and thus
\[|b_i|_{t+1} \leq 4 - \Phi(i, t+1).\]

\paragraph{Case 2: \boldmath$j = i$.}

In this case, we know $b_i$ was just chopped down, so 
\begin{align*}
|b_i|_{t+1} = 0 &< 2 \\ &\leq 4 - \Phi(i, t+1).
\end{align*}
Here we use the fact that $\Phi$ never exceeds $2$ as
\[\Phi(i, t+1) \leq \sum_{k \in [n]} 2h_k \leq 2.\]

\paragraph{Case 3: \boldmath$j > i$.}

In this case, a slower-growing bamboo, $b_j$, was cut down in the stead of $b_i$.
We know that $b_j$, with height $|b_j|_t + h_j$, is the tallest bamboo during the intermediate step between $t$ and $t+1$. So
\begin{equation} \label{eq:case2b_i}
|b_i|_{t+1} \leq |b_j|_{t} + h_j.
\end{equation}

How does $V$ change from $V(j, t)$ to $V(i, t+1)$? During the growth phase, $\sum_{k=1}^i h_k \leq 1 - h_j$ units of volume are added to the bamboo with indices $1, \ldots, i$. On the other hand, $V(j, t)$ includes at least $2 - h_j$ units of volume from bamboo $b_j$, which $V(i, t+1)$ does not include. Thus we have
\begin{equation}
V(i, t+1) \leq V(j, t) - 1.
\label{eq:V_change}
\end{equation}
Each unit of volume is weighted by at least $h_j$ in both $\Phi(j, t)$ and $\Phi(i, t+1)$, so by Equation \eqref{eq:V_change} we have
\begin{equation} \label{eq:case2excess}
\Phi(i, t+1) \leq \Phi(j, t) - h_j.
\end{equation}

We know by the inductive hypothesis that
\[|b_j|_t \leq 4 - \Phi(j, t).\]
Substituting by Equations \eqref{eq:case2b_i} and \eqref{eq:case2excess} we have
\[(|b_i|_{t+1} - h_j) \leq 4 - (\Phi(i, t+1) + h_j)\]
and thus
\[|b_i|_{t+1} \leq 4 - \Phi(i, t+1).\]

\end{proof}

We conclude the section by proving Theorem \ref{thm:greedy}.

\begin{proof}[Proof of Theorem \ref{thm:greedy}]

  It follows from Lemma \ref{lem:greedy} that no bamboo can achieve height $4$ even during an intermediate step. Recall that bamboo $i$ has height $|b_i|_t + h_i$ during the intermediate step after time $t$. If $|b_i|_t < 2$, then it follows immediately that $|b_i|_t + h_i < 3 < 4$. Otherwise, we know $V(i, t) \geq 2$, which implies that $\Phi(i, t) \geq 2h_1$. Thus we can apply Lemma \ref{lem:greedy} to find the bound

  \begin{align*}
    |b_i|_t + h_i &\leq 4 - \Phi(i, t) + h_i \\
    &\leq 4 - 2h_1 + h_i \\
    &\leq 4 - h_1 < 4.
  \end{align*}

Note that we have, in fact, proved a slightly stronger claim than that of Theorem \ref{thm:greedy}. Not only does \textbf{Reduce-Max} limit the backlog to $4$, it actually limits the backlog to $4 - h_1$, i.e., $4$ minus the largest growth rate among the bamboo.
\end{proof}

\section{Reduce-Fastest}

\textbf{Reduce-Fastest($x$)} is a bamboo trimming algorithm that cuts down the fastest-growing bamboo with height at least $x$ at each time step (if no bamboo has height at least $x$, then no action is taken). 

\textbf{Reduce-Fastest($x$)} was first studied by G\k{a}sieniec, Klasing, et al. in the case of $x = 2$ in \cite{GasieniecKlLe17}. They proved that \textbf{Reduce-Fastest($2$)} achieves backlog $4$. In \cite{DEmidioStNa19}, D'Emidio et al. performed an extensive experimental evaluation of several bamboo garden trimming algorithms including \textbf{Reduce-Fastest($1$)} and \textbf{Reduce-Fastest($2$)}. The authors conjectured that \textbf{Reduce-Fastest($1$)} achieves backlog $2$ and \textbf{Reduce-Fastest($2$)} achieves backlog $3$. We are able to disprove the first conjecture, and prove a more general form of the second conjecture: that \textbf{Reduce-Fastest}($x$) limits the backlog to $x+1$ for all $x \geq 2$, and that this bound is tight.

\begin{theorem}
    For all $x \geq 2$, \textbf{Reduce-Fastest}($x$) prevents any bamboo from achieving height $x+1$. (And thus the backlog is strictly less than $x+1$.)
    \label{thm:reduce_fastest}
\end{theorem}

\begin{proof}
    Suppose for the sake of contradiction that we have $n$ bamboo $b_1, \ldots, b_n$ with corresponding fill-rates $h_1, \ldots, h_n$, and that some bamboo $b_i$ achieves height $x+1$ at time $t_3$ after most-recently achieving height $x$ at time $t_1$. We then consider the bamboo that are cut at least once in $[t_1, t_3)$, and denote the set of such bamboo by $S$. For all $b_j \in S$, we denote by $m_j$ the number of times that $b_j$ is cut in the interval $[t_1, t_3)$. 
    
    The following claim shows that for all $b_j \in S$, $h_j \geq m_j \cdot h_i$. That is, for a bamboo to be cut down $m$ times in the interval $[t_1, t_3)$, it must have fill-rate at least $m$ times that of $b_i$.
    
    \begin{claim}
    For all $b_j \in S$, we have $h_j \geq m_j \cdot h_i$.
    \label{clm:fill_rate_budgeting}
    \end{claim}

    \begin{proof}
    We begin by considering the case of $m_j = 1$. In this case we have $h_j \geq h_i$, as $h_j$ was cut down by \textbf{Reduce-Fastest}($x$) at a time when $b_i$ had height at least $x$, so $b_j$ must be at least as fast-growing as $b_i$.
    
    Next we consider the case of $m_j \geq 2$. In this case we know $b_j$ is cut down $m_j$ times in the interval $[t_1, t_3)$, so it must grow at least $x(m_j - 1)$ in that interval as $b_j$ has to regrow to a height of at least $x$ units between successive cuts. However, $b_i$ fails to grow even $1$ unit during the same interval $[t_1, t_3)$, as it has height at least $x$ at time $t_1$, and it has not yet achieved height $x+1$ at time $t_3 - 1$. In other words, in the time that bamboo $b_j$ grows by at least $x(m_j - 1)$, bamboo $b_i$ grows by less than $1$. Thus 
    \begin{align*} 
    h_j & \geq x(m_j - 1) h_i\\ & \geq 2(m_j - 1) h_i \\ & \geq m_j h_i,\end{align*}
    where the final inequality uses $m_j \geq 2$.
\end{proof}

    We now consider the length of the interval $[t_3, t_1)$. We have by Claim \ref{clm:fill_rate_budgeting} that
    
    \begin{align*}
    t_3 - t_1 &= \sum_{b_j \in S} m_j \\
              &\leq \sum_{b_j \in S} \frac{h_j}{h_i} \\
              &\leq \frac{1}{h_i} \sum_{b_j \in S} h_j \\
              &\leq \frac{1}{h_i} (1 - h_i) \\
              &= 1/h_i - 1.
    \end{align*}
    
    Moreover $t_3 - t_1$ is integer, so
    
    \begin{align*}t_3 - t_1 &\leq \lfloor 1/h_i - 1 \rfloor \\ & \leq \lfloor 1/h_i \rfloor - 1.\end{align*}
    
    But this means that the interval is too short for $b_i$ to reach height $x+1$. In particular, $b_i$ requires at least $\lfloor 1/h_i \rfloor$ time to achieve height $x+1$ after achieving height $x$. To be explicit, we have that at time $t_1$ bamboo $b_i$ has height strictly less than $x + h_i$, and thus that at time $t_3$ bamboo $b_i$ has height strictly less than 
    \begin{align*}x + h_i + (t_3 - t_1)h_i & \leq x + h_i(1 + \lfloor 1/h_i \rfloor - 1) \\ & \leq x+1.\end{align*}
    Therefore $b_i$ does not achieve height $x+1$ at time $t_3$, which is a contradiction.
    
\end{proof}

The bound of $x+1$ for all $x \geq 2$ on the backlog guaranteed by \textbf{Reduce-Fastest}$(x)$ is tight. We believe this is the first tight bound on \textbf{Reduce-Fastest}$(x)$ for any value of $x$.

\begin{theorem}
For any $\varepsilon, x > 0$, there exists some $n \in \mathbb{N}$ such that \textbf{Reduce-Fastest}$(x)$ allows for backlog at least $x + 1 - \varepsilon$.
\label{thm:reduce_fastest_general_lower_bound}
\end{theorem}

\begin{proof} 
Consider $n$ bamboo with uniform fill rates $h_i = 1/n$ for all $i$. No bamboo will be cut down until they all simultaneously achieve height at least $x$. Then over the next $n$ time steps, all of the bamboo will be cut down, with the last bamboo reaching height at least $x + (n-1)/n$.

Setting $n = \lceil 1 / \varepsilon \rceil$, we obtain a backlog at least $$x + 1 - 1/n \geq x + 1 - \varepsilon.$$
\end{proof}

We conclude this section by providing lower bounds on the backlog achieved by $\textbf{Reduce-Fastest}(x)$. In particular, we give a counterexample to the conjecture that \textbf{Reduce-Fastest}$(1)$ achieves backlog $2$. Interestingly, it remains an open problem as to whether \textbf{Reduce-Fastest}$(1)$ achieves any finite backlog.

\begin{theorem}
     \textbf{Reduce-Fastest}$(1)$ does not achieve any backlog less than $3$.
\end{theorem}

\begin{proof}
    Suppose we have $f$ fast bamboo with growth rates $$1/(f + \sqrt{f})$$ and $s = \sqrt{f} + 1$ fast bamboo with growth rates $$1 / (f + 2\sqrt{f} + 2)$$ for some $f \in \mathbb{N}$ that is a perfect square. 
    
    We note that in this construction, 
    \begin{align*} \sum_i h_i &= f(f + \sqrt{f}) + (\sqrt{f} + 1) / (f + 2\sqrt{f} + 2) \\ & < \sqrt{f}(\sqrt{f} + 1) + 1/(\sqrt{f} + 1) \\ & = 1.
    \end{align*}
    Thus the sum of the fill rates is less than $1$, and so this is a valid construction of bamboo.
    
    Now we examine the behavior these bamboo exhibit when the cutting player utilizes the \textbf{Reduce-Fastest}$(1)$ strategy. Initially, all bamboo have height less than $1$, so the player does not cut any bamboo down. At time $f + \sqrt{f}$, the fast bamboo all simultaneously achieve height $1$. Thus at time steps $$ f + \sqrt{f}, \ldots, 2f + \sqrt{f} - 1,$$ the $f$ fast bamboo are cut down. Then all of the fast bamboo have height less than $1$, and the slow bamboo have all achieved height $1$. Thus $\textbf{Reduce-Fastest}(1)$ will cut down the slow bamboo until a fast bamboo again achieves height $1$. Therefore during each of the time steps $$2f + \sqrt{f}, \ldots, 2f + 2\sqrt{f} - 1$$ a slow bamboo will be cut down. During those time steps, $\sqrt{f} = s-1$ slow bamboo are cut down, meaning exactly $1$ slow bamboo has not yet been cut by time $2f + 2\sqrt{f}$. The \textbf{Reduce-Fastest}$(1)$ algorithm does not have time to cut this last slow bamboo, as at time $2f + 2\sqrt{f}$ the first of the fast bamboos that was cut again achieves height $1$. Since $\textbf{Reduce-Fastest}$ prioritizes fast bamboo, it will then cut down each of the fast bamboo during time steps $$2f + 2\sqrt{f}, \ldots, 3f + 2\sqrt{f} - 1$$ as they successively achieve height $1$. Thus the final remaining uncut slow bamboo will not be cut for the first time until at least time $3f + 2\sqrt{f}$, by which time it has achieved height \begin{align*}
        &(3f + 2\sqrt{f}) / (f + 2\sqrt{f} + 2) \\ & = 3 - (4\sqrt{f} + 6) / (f + 2\sqrt{f} + 2) \\ &= 3  - o(1).
    \end{align*}
    
    Thus this construction of bamboo achieves backlog arbitrarily close to $3$ as $f \to \infty$.
\end{proof} 

The following theorem, while rather simple, serves to show that $\textbf{Reduce-Fastest}(x)$ cannot achieve worst-case backlog arbitrarily close to the optimal value of $2$ for any value of $x$.

\begin{theorem}
There is no value of $x$ for which \textbf{Reduce-Fastest}$(x)$ achieves backlog $2.01$.
\end{theorem}

\begin{proof}
By theorem \ref{thm:reduce_fastest_general_lower_bound}, we know that this holds for any value of $x > 1.01$. Also note for $0 < x < 1$, a simple construction with one bamboo of growth rate $x$ and another of growth $1 -x$ achieves infinite backlog, as the slower of the two bamboos is never cut down. (And for $x \leq 0$, any construction with multiple bamboo with distinct growth rates achieves unbounded backlog).

We now offer a simple, concrete construction that holds for any value of $1 \leq x \leq 1.01$. Suppose we have $900$ bamboo with growth rates $1/ 1000$ and $140$ bamboo with growth rates $1/1400$. The sum of the growth rates is $\sum_i h_i = 1$. This construction is very similar to the construction of fast and slow bamboo offered in the preceding theorem, but it is loose enough to continue to offer a bound on the backlog of at least $2.01$ for all $x < 1.01$.

The fast bamboo achieve height $x$ somewhere between time $1000$ and time $1010$ depending on $x$. Then each of the $900$ fast bamboo are cut down, and then some proper subset of the slow bamboo are cut down until the first cut of the fast bamboo again achieves height $x$ at some time no later than $2020$. Then the fast bamboo are again all cut. By the time that all of the fast bamboo have been cut exactly twice, at least $2900$ time steps have elapsed, and some slow bamboo remains uncut with height at least $2900 / 1400 > 2.01$.
\end{proof}

Thus, we can eliminate $\textbf{Reduce-Fastest}(x)$ from consideration in the search for a simple, optimal bamboo-cutting algorithm.

\vspace{1 cm} 

\section{Deadline-Driven Strategy}

We now present a very simple algorithm, the \textbf{Deadline-Driven Strategy}, which achieves backlog $2$. The algorithm, which is novel in the context of bamboo trimming, was introduced by Liu and Layland \cite{LiuLa73} in the context of a related scheduling problem in the early $70$'s.

We begin by translating the result of Liu and Layland \cite{LiuLa73} to be in terms of the fixed-rate cup game, which we remind the reader is defined as follows. There are $n$ cups with fill rates $h_1, \ldots, h_n$ satisfying $\sum_{i=1}^n h_i \leq 1$. At the beginning of each time step, each cup $i$ receives $h_i$ units of water. The player then selects a cup from which to remove $1$ unit of water --- if the cup contains less than $1$ unit of water, it is emptied. The backlog for the cup game is defined as the height of the fullest cup.

The Deadline Driven Strategy examines all cups with height at least $1$, and removes water from the cup that will soonest reach height $2$ --- it arbitrarily chooses from the cups with the closest deadline of reaching height $2$.

One interpretation of Theorem $7$ from Liu and Layland's paper \cite{LiuLa73} is that the \textbf{Deadline-Driven Strategy} achieves backlog 2 for the fixed-rate cup game (i.e., the non-flushing bamboo game) as long as the fill-rates are inverse integers --- each $h_i$ is equal to $1/k_i$ for some $k_i \in \mathbb{N}$. Subsequent work \cite{LitmanMo09} rediscovered the \textbf{Deadline-Driven Strategy} for a related scheduling problem; one consequence of their arguments is that, \emph{if backlog 2 is possible}, then the \textbf{Deadline-Driven Strategy} achieves it. Since \cite{BaruahCoPl96} established that backlog 2 is, in fact possible (using results from network flow theory), it follows that one can remove the inverse-integer constraint on Liu and Layland's result \cite{LiuLa73}. That is, the \textbf{Deadline-Driven Strategy} achieves backlog 2 in the fixed-rate cup game for any set of fill rates.

We now give an alternative analysis of the \textbf{Deadline-Driven Strategy} that applies to both the fixed-rate cup game and the bamboo trimming problem---our analysis is significantly simpler than those in past work. 

\begin{theorem}
Suppose that we have $n$ cups $b_1, \ldots, b_n$ with corresponding fill-rates $h_1, \ldots, h_n$ satisfying $\sum_{i = 1}^n h_i \leq 1$. Then the \textbf{Deadline-Driven Strategy} for the cup game will limit the backlog to strictly less than $2$. Furthermore, the equivalent strategy will also limit the backlog to strictly less than $2$ if applied to $n$ bamboo with fill-rates $h_1, \ldots, h_n$.
\end{theorem}

\begin{proof}
Following the terminology from Liu and Layland's paper, we say that cup $b_i$ is \defn{requested} at time $t$ if it reaches height $1$ at time $t$. At a time $t_0$, we say that a cup $b_i$ has a \defn{deadline} at time $t$ if $|b_i|_{t_0} \geq 1$ and $|b_i|_{t_0} + (t - t_0)h_i \in [2, 2+h_i)$, that is, the cup has height at least $1$ and it will achieve height $2$ at time $t$ if no water is removed during the interval $[t_0, t)$. We say that cup $b_i$ \defn{overflows} at time $t$ if it achieves height $2$ at time $t$, i.e., the cup is not attended to before its deadline. Whenever a cup with fill in the range $[1, 2)$ is emptied from, we say that the request (when the cup previously reached height $1$) is \defn{completed}.

Suppose for the sake of contradiction that cup $i$ overflows at time $t_3$ and that this is the first ever overflow. We then define $t_1$ to be the last time prior to $t_3$ during which either the player is idle (as all cups have height strictly less than $1$) or the player chooses a cup with deadline after $t_3$. 

We consider the time interval $(t_1, t_3) = [t_1 + 1, t_3 -1]$, during which we know the player is busy removing water from cups with deadlines at or before $t_3$. Furthermore, at time $t_1$, the player was either idle or was busy removing water from a cup with deadline after $t_3$. Thus all cups with request time $\leq t_1$ and deadline $\leq t_3$ had already had their requests completed by time $t_1$. Thus in the interval $(t_1, t_3)$, the player is continuously working on tasks with request time after $t_1$ and deadline at or before $t_3$. 

Now we count the number of requests that occur at or after $t_1+1$ with deadline at or before $t_3$. We call such requests and deadlines \defn{critical}. 

\paragraph{Case 1: There are at least \boldmath$t_3 - t_1$ critical deadlines.} Let $\varepsilon_j = 1 - |b_j|_{t_1}$ for each cup $j$. Since every request at or before $t_1$ with a deadline in $(t_1, t_3]$ is completed by time $t_1$, we know that any cup at time $t_1$ that has fill $1$ or larger must have a deadline after $t_3$ and must not contribute any critical requests/deadlines. Thus each cup $j$ that has at least one critical deadline satisfies $\varepsilon_j > 0$. It follows that if a cup $j$ has $r > 0$ critical deadlines, then the total amount of water placed into cup $j$ during times $(t_1, t_3]$ is at least $r + \epsilon_j > r$. Since at least $t_3 - t_1$ total critical deadlines occur in the interval $(t_1, t_3]$, it follows that more than $t_3 - t_1$ water is placed into cups during those $t_3 - t_1$ steps, a contradiction.

\paragraph{Case 2: There are at most \boldmath$t_3 - t_1 - 1$ critical deadlines.} Since the player is non-idle during the interval $(t_1, t_3)$, and since $t_3$ is the first step during which any cup overflows, the player completes $t_3 - t_1 - 1$ critical requests, one during each of the steps $t_1+1, \ldots, t_3-1$.  Additionally, the final deadline for cup $i$ (which overflows at time $t_3$) is not met and thus corresponds to a critical request that is not completed. Hence there are at least $t_3 - t_1 - 1$ critical requests that get completed during the interval $(t_1, t_3]$ and at least $1$ critical request that does not get completed. This contradicts the assumption that there are $t_3 - t_1 - 1$ or fewer critical requests/deadlines.

\vspace{.3 cm}

Precisely the same analysis that we have used to prove the theorem for the fixed-rate cup game also applies to the bamboo-garden trimming problem. (Indeed, the bamboo-garden trimming problem can be modelled as a version of the fixed-rate cup game in which the player empties cups instead of only removing $1$ unit of water.)
\end{proof}

Thus, the \textbf{Deadline-Driven Strategy} is a simple algorithm which achieves the optimal worst-case backlog of $2$.

\begin{acks}
The author would like to thank Michael A. Bender and William Kuszmaul for their mentorship and advice during this project.
We gratefully acknowledge support from National Science Foundation grants CNS-1938709 and CCF-2106827.
\end{acks}


\bibliographystyle{siam}
\balance
\bibliography{all}

\end{document}